\documentclass{article}
\usepackage[utf8]{inputenc}

\usepackage{amsmath,amssymb,bbm,amsthm}
\usepackage{fullpage}
\usepackage{thm-restate,color,xcolor,xspace}
\usepackage{hyperref,cleveref}
\usepackage{algorithm}
\PassOptionsToPackage{noend}{algpseudocode}
\usepackage{algpseudocode}
\usepackage{thm-restate}
\usepackage{graphicx}

\algrenewcommand\algorithmicrequire{\textbf{Input:}}

\usepackage{amsmath}
\newcommand*\diff{\mathop{}\!\mathrm{d}}

\newtheorem{theorem}{Theorem}[section]
\newtheorem{lemma}[theorem]{Lemma}

\newtheorem{definition}[theorem]{Definition}
\newtheorem{corollary}[theorem]{Corollary}

\newtheorem{claim}[theorem]{Claim}

\newcommand{\E}{\mathsf{E}}

\title{The wrong direction of Jensen's inequality is algorithmically right}
\author{Or Zamir\\ \emph{Princeton University}}
\date{}

\begin{document}

\maketitle

\begin{abstract}
Let~$\mathcal{A}$ be an algorithm with expected running time~$e^X$, conditioned on the value of some random variable~$X$.
We construct an algorithm~$\mathcal{A'}$ with expected running time~$O\left(e^{\E[X]}\right)$, that fully executes~$\mathcal{A}$.
In particular, an algorithm whose running time is a random variable~$T$ can be converted to one with expected running time~$O\left(e^{\E[\ln T]}\right)$, which is never worse than~$O(\E[T])$.
No information about the distribution of~$X$ is required for the construction of~$\mathcal{A}'$.
\end{abstract}

\section{Introduction}

Let~$\mathcal{A}$ be a \emph{Las Vegas}\footnote{A randomized algorithm is called \emph{Las Vegas} if it always returns the correct answer, but its running time is a random variable.} randomized algorithm.
Assume that conditioned on the value of some random variable~$X$, the expected running time of~$\mathcal{A}$ is~$e^X$.
By Jensen's inequality,~$\E[e^X]\geq e^{\E[X]}$, and in fact~$\mathcal{A}$'s expected running time might be much larger than $e^{\E[X]}$: Consider for example~$X$ that gets the value~$\frac{1}{p}\E[X]$ with probability~$p$ and~$0$ otherwise, for any choice of~$p>0$;
While the expectation of~$X$ is always~$\E[X]$, the expectation of~$e^X$ is~$p\cdot e^{\frac{1}{p}\E[X]}$ which can be arbitrarily large.
We show that, surprisingly, any such~$\mathcal{A}$ can be converted to a different Las-Vegas randomized algorithm~$\mathcal{A'}$ that gives the same answer yet runs in expected time~$O\left(e^{\E[X]}\right)$.
Transforming~$\mathcal{A}$ to~$\mathcal{A'}$ does not require any assumption or knowledge about the distributions of~$X$.

\begin{theorem}\label{mainthm}
There exists an algorithm~$T$ that receives as an input a randomized \emph{Las Vegas} algorithm~$\mathcal{A}$, and fully executes it.
If the expected running time of~$\mathcal{A}$ is~$e^X$ when conditioned on the value of some random variable~$X$, then the expected running time of~$T(A)$ is~$O\left(e^{\E[X]}\right)$.
\end{theorem}

As a corollary, any algorithm whose running time is a random variable~$T$ can be converted to one with expected running time~$O\left(e^{\E[\ln T]}\right)$, which is never worse than~$O(\E[T])$.

Recently we used the following simple version of Theorem~\ref{mainthm} in a late revision of~\cite{zamir2021faster} to substantially simplify the analysis in the paper. 
The paper improves the running time of exact exponential-time algorithms for general Constraint Satisfaction Problems.
\begin{lemma}[from an up-to-date version of~\cite{zamir2021faster}]\label{simple}
Let~$\mathcal{A}$ be an algorithm with expected running time~$2^X$ conditioned on the value of a random variable~$X$.
There exists an algorithm~$\mathcal{A'}$ that fully executes~$\mathcal{A}$ and has an expected running time of~$O\left(2^{\E\left[X\right]}\cdot \E\left[X\right]\right)$.
Transforming~$\mathcal{A}$ to~$\mathcal{A}'$ requires knowing~$\E[X]$.
\end{lemma}

In this paper we focus on Theorem~\ref{mainthm} itself, obtaining an optimal version of it.

We transform an algorithm~$\mathcal{A}$ by using a sequence of \emph{truncated evaluations}.
A \emph{truncated evaluation} of an algorithm~$\mathcal{A}$ for~$t$ steps is the process of running algorithm~$\mathcal{A}$ and aborting its run if it did not fully execute in the first~$t$ computational steps of its run.
Each of the algorithms we present is thus a sequence of values~$t_1,t_2,\ldots,t_i,\ldots$ which we use as thresholds for truncated evaluations of~$\mathcal{A}$.
We stop at the first time~$\mathcal{A}$ is fully executed.
These thresholds can be defined deterministically or be random variables. 
In the simpler algorithms we present, the thresholds depend on~$\E[X]$ or even on the entire distribution~$X$. For the proof of Theorem~\ref{mainthm} the thresholds are completely independent of~$X$ and~$\mathcal{A}$.

Truncated evaluations are frequently used in complexity theory (for example, see the proof of the time and space hierarchies in~\cite{arora2009computational}).
The first algorithmic use of such a sequence of truncated evaluations that we are aware of, is by Alt, Guibas, Mehlhorn, Karp and Wigderson \cite{alt1996method}.
They used it to convert \emph{Las Vegas} randomized algorithms to \emph{Monte Carlo} randomized algorithms, with success probability larger than what Markov's inequality gives. 
Luby, Sinclair and Zuckerman~\cite{luby1993optimal} then introduced a universal strategy for truncated evaluations. That is a sequence that is guaranteed to run in time~$O(s \log s)$ if there exists any sequence that runs in time~$O(s)$ for the same algorithm.
Our contribution thus is two-fold: first, we prove the existence of good strategies in terms of~$\E[X]$, and second, we show that these strategies can be explicitly constructed (i.e., without paying additional logarithmic factors).
Not paying additional factors guarantees, due to Jensen's inequality, that our transformed algorithm is never worse than the original algorithm.

A natural use for such theorems is the regime of exponential-time algorithms.
For example, consider the PPSZ algorithm for solving~$k$-SAT, its recent improvements, and generalizations for CSPs \cite{paturi2005improved} \cite{hertli20143} \cite{hansen2019faster} \cite{zamir2021faster}.
In these algorithms, a randomly chosen permutation determines the number of input variables that we need to \emph{guess} the values of. The expectation of this number of variables is then analyzed.
The success probability or running time is exponential in this number.
In the original PPSZ algorithm the analyzed quantity is the success probability and thus Jensen's inequality is applicable to bound this probability from below.
In other variations (including~\cite{zamir2021faster}), the analyzed quantity is the running time and then Jensen's inequality is no longer applicable and either a more complicated analysis or the statement of this paper is necessary.
Further discussion on possible applications and in particular possible implications for SAT algorithms appears in Section~\ref{sec:conclude}.

\subsection{Preliminaries}\label{prelim}
We use standard notation throughout the paper.
The notation $\ln x$ is used for the natural logarithm, and~$\log x$ is used for the base two logarithm.
\begin{definition}[Iterated functions]
Let~$f:\mathbb{R}\rightarrow\mathbb{R}$ be a function.
We define the iterated functions~$f^{(k)}:\mathbb{R}\rightarrow\mathbb{R}$ recursively as follows.
$f^{(0)}(x):=x$, and for any~$k>0$ we let~$f^{(k)}(x):=f\left(f^{(k-1)}\left(x\right)\right)$.
\end{definition}
\begin{definition}[Star functions]
Let~$f:\mathbb{R}\rightarrow\mathbb{R}$ be a function.
Assume~$f$ is strictly increasing and strictly shrinking\footnote{That is,~$f(x)<x$.} for all~$x\geq x_0$.
The star function of~$f$, defined with respect to~$x_0$ for every~$x\geq x_0$, is
$$
f^\star(x) = \min\{k \;|\; f^{(k)}(x)\leq x_0\}.
$$
\end{definition}

The (general) Tower function~$Tower_b\left(n,\; x\right) : \mathbb{N}\times\mathbb{R}\rightarrow\mathbb{R}$ is defined as~$f^{(n)}(x)$ where~$f(x)=b^x$.
The standard Tower function~$Tower:\mathbb{N}\rightarrow\mathbb{N}$ is defined as~$Tower(n)=Tower_2(n,1)$.
The discrete inverse of the Tower function is~$\log^\star$, defined with respect to~$x_0=2$. 
That is,~$\log^\star n$ is the smallest integer such that~$Tower(\log^\star n)\geq n$.

\section{Proof of Theorem~\ref{mainthm}}
We begin by presenting a simple proof of Lemma~\ref{simple}.

Let~$\mathcal{A}$ be an algorithm whose expected running time is~$e^X$ when we condition on the value of some non-negative random variable~$X$.
We observe, by Markov's inequality, that
$$
\Pr\left(X > \E[X]+1\right)\leq \frac{\E[X]}{\E[X]+1} = 1-\frac{1}{\E[X]+1}.
$$

Hence, consider the following algorithm.
\begin{algorithm}[H]
\caption{}\label{algbasic}
\begin{algorithmic}[1]
\Require{$\mathcal{A}$, $\E[X]$.}
\Repeat
\State Run $\mathcal{A}$ for~$2e^{\E\left[X\right]+1}$ computational steps.
\Until{$\mathcal{A}$ completed a run.}
\end{algorithmic}
\end{algorithm}

\begin{lemma}
Algorithm~\ref{algbasic} is expected to terminate in~$O\left(e^{\E\left[X\right]}\cdot \E\left[X\right]\right)$ time.
\end{lemma}
\begin{proof}
If~$X\leq \E[X]+1$ then the expected running time of~$\mathcal{A}$ is at most~$e^{\E[X]+1}$, and hence by Markov's inequality a truncated evaluation of~$\mathcal{A}$ for~$2e^{\E\left[X\right]+1}$ steps concludes with probability at least~$\frac{1}{2}$.
By another application of Markov's inequality we got $\Pr\left(X \leq \E[X]+1\right)>\frac{1}{\E[X]+1}$.
Hence, the expected number of iterations until the truncated evaluations concludes is at most~$2\left(\E\left[X\right]+1\right)$. 
Each iteration takes~$O\left(e^{\E\left[X\right]}\right)$ time.
\end{proof}

\subsection{Optimal bound when the distribution of~$X$ is known}
The bound given by Markov's inequality in
$$
\Pr\left(X \geq \E[X]+1\right)\leq \frac{\E[X]}{\E[X]+1} = 1-\frac{1}{\E[X]+1}
$$
is attained \emph{only} by the following distribution of~$X$:
\[
\Pr\left(X=k\right) :=
\begin{cases} 
      \frac{1}{\E[X]+1} & k=0 \\
      1-\frac{1}{\E[X]+1} & k=\E[X]+1
   \end{cases}.
\]
In this distribution, on the other hand, the value of~$X$ is very small with a relatively high probability.
In particular, in the case where~$X<\E[X]+1$ we need to run~$\mathcal{A}$ for much less than~$e^{\E[X]+1}$ computational steps.
Hence, it is sensible to hope that every distribution~$X$ has some threshold other than~$\E[X]+1$ for which an algorithm similar to Algorithm~\ref{algbasic} results in a better bound.
Consider the following algorithm, which is a generalization of Algorithm~\ref{algbasic} in which the threshold can be arbitrary.

\begin{algorithm}[H]
\caption{}\label{algbasict}
\begin{algorithmic}[1]
\Require{$\mathcal{A}$, $t$.}
\Repeat
\State Run $\mathcal{A}$ for~$2e^{t}$ computational steps.
\Until{$\mathcal{A}$ completed a run.}
\end{algorithmic}
\end{algorithm}

\begin{lemma}\label{knownX}
Let~$X$ be a non-negative random variable.
There exists~$t\in [0,\E[X]+1]$ such that~$\frac{e^t}{\Pr\left(X< t\right)} \leq e^{\E[X]+1}$.
\end{lemma}
\begin{proof}
Assume by contradiction that $\frac{e^t}{\Pr\left(X< t\right)} > e^{\E\left[X\right]+1}$ for every~$t\in [0,\E[X]+1]$.
Equivalently, 
$$
\Pr\left(X\geq t\right) = 
1 - \Pr\left(X< t\right) >
1 - e^{t-\left(\E\left[X\right]+1\right)}
.
$$
Therefore,
\begin{align*}
    \E[X] &= \int_{0}^{\infty} \Pr(X\geq t) \diff t \geq \int_{0}^{\E[X]+1} \Pr(X\geq t) \diff t \\
    &> \int_{0}^{\E[X]+1} \left(1 - e^{t-\left(\E\left[X\right]+1\right)}\right) \diff t\\
    &= \left(\E[X]+1\right) - \left(1 - e^{-\left(\E\left[X\right]+1\right)}\right) 
    = \E[X] + e^{-\left(\E\left[X\right]+1\right)} \\&> \E[X],
\end{align*}
which is a contradiction.
\end{proof}

Lemma~\ref{knownX} implies the following.
\begin{corollary}
For every distribution~$X$ there exists a value of~$t=t(X)$ for which Algorithm~\ref{algbasict} runs in~$O\left(e^{\E[X]}\right)$ time.
\end{corollary}

We note that the additive constant~$+1$ in the exponent in Lemma~\ref{knownX} is necessary. 
For a parameter~$E$, consider the random variable~$X$ supported on~$[0,E+1+\ln \left(1+e^{-\left(E+1\right)}\right)]$ and distributed with density~$f(x) := e^{x-\left(E+1\right)}$;
Its expectation is
\begin{align*}
\E[X] &= \int_0^{E+1+\ln \left(1+e^{-\left(E+1\right)}\right)} xf(x) \diff x \\
&= \left(\left(x-1\right)e^{x-\left(E+1\right)}\right)\bigg\rvert_{0}^{E+1+\ln \left(1+e^{-\left(E+1\right)}\right)} \\
&= \left(E + \ln \left(1+e^{-\left(E+1\right)}\right)\right)\cdot\left(1 + e^{-\left(E+1\right)}\right) + e^{-\left(E+1\right)}\\
&= E + O\left(e^{-\left(E+1\right)}\right) = E+o(1)
,
\end{align*}
where the~$o(1)$ term is vanishing when~$E\rightarrow \infty$.
On the other hand, for any~$t\geq 0$ we have
$$
\frac{e^t}{\Pr\left(X< t\right)} = 
\frac{e^t}{\min\left(1,\;e^{t-\left(E+1\right)} - e^{-\left(E+1\right)}\right)}
>
\frac{e^t}{e^{t-\left(E+1\right)}} = e^{E+1}.
$$

\subsection{Optimal algorithm when the distribution of~$X$ is unknown}
If the only thing known about the distribution of~$X$ is its expectation~$\E[X]$, then there is no fixed value of~$t$ for which Algorithm~\ref{algbasict} is better than Algorithm~\ref{algbasic}. 
Fix a value of~$\E[X]$ and a choice of~$t$.
If~$t<\E[X]$ then with the constant distribution~$X\equiv \E[X]$ Algorithm~\ref{algbasict} never terminates.
Otherwise,~$t\geq \E[X]$ and we consider the following distribution~$X$:
\[
\Pr\left(X=k\right) :=
\begin{cases} 
      1-\frac{\E[X]}{t+1} & k=0 \\
      \frac{\E[X]}{t+1} & k=t+1
   \end{cases}.
\]
For this distribution, the expected running time of Algorithm~\ref{algbasict} is 
$$
\frac{e^t}{1 - \frac{\E[X]}{t+1}}
=
e^{\E[X]}\cdot \frac{e^s}{\left(\frac{s+1}{\E[X]+s+1}\right)}
=
e^{\E[X]} \left(1 + \frac{\E[X]}{s+1}\right)e^s
\geq
e^{\E[X]} \E[X] \cdot \frac{e^s}{s+1} \geq e^{\E[X]} \E[X]
,
$$
where~$s:=t-\E[X]\geq 0$ and the last inequality follows as~$e^s\geq s+1$ for any~$s$. 

To improve Algorithm~\ref{algbasic} then, we need to consider \emph{several thresholds}.
We demonstrate this idea with the following Lemma.
\begin{lemma}\label{twothresholds}
Let~$X$ be a non-negative random variable.
It holds that either~$\Pr\left(X\leq \E\left[X\right] - \ln \E\left[X\right]\right) > \frac{1}{\E\left[X\right]+1}$ or~$\Pr\left(X\leq \E\left[X\right] + 2\right) > \frac{1}{\ln\E\left[X\right]+2}$.
\end{lemma}
\begin{proof}
Assume that~$p:=\Pr\left(X\leq \E\left[X\right] - \ln \E\left[X\right]\right) \leq \frac{1}{\E\left[X\right]+1}$.
We observe that
\begin{align*}
\E\left[X\right] &= 
p\; \E\left[X \;|\; X \leq \E\left[X\right]-\ln \E\left[X\right]\right]
+
\left(1-p\right)\E\left[X \;|\; X > \E\left[X\right]-\ln \E\left[X\right]\right]
\\
&\geq \left(1-p\right)\E\left[X \;|\; X > \E\left[X\right]-\ln \E\left[X\right]\right]
,
\end{align*}
and hence
\begin{align*}
\E\left[X \;|\; X > \E\left[X\right]-\ln \E\left[X\right]\right] &\leq
\frac{1}{1-p}\E[X]\\
&\leq \frac{1}{1-\frac{1}{\E[X]+1}}\E[X]\\
&=\E[X]+1
.
\end{align*}
Denote by~$Y:=X-\left(\E\left[X\right]-\ln \E\left[X\right]\right)$.
The above can now be rephrased as~$\E\left[Y \; | \; Y>0\right] \leq \ln \E\left[X\right] + 1$.
Applying Markov's inequality to~$Y$ conditioned on $Y>0$, we get
$$
\Pr\left(Y > \ln \E\left[X\right] + 2 \; | \; Y>0\right) \leq \frac{\ln \E\left[X\right] + 1}{\ln \E\left[X\right] + 2} = 1 - \frac{1}{\ln \E\left[X\right] + 2}
.
$$
We conclude by noting that
$$
\Pr\left(X > \E\left[X\right] + 2\right)
=
\Pr\left(Y > \ln \E\left[X\right] + 2\right)
\leq
\Pr\left(Y > \ln \E\left[X\right] + 2 \; | \; Y>0\right)
.
$$
\end{proof}

Consider the following Algorithm.
\begin{algorithm}[H]
\caption{}\label{alg2thresh}
\begin{algorithmic}[1]
\Require{$\mathcal{A}$, $\E[X]$.}
\Repeat
\For{$\lceil \E[X]+1 \rceil$ \text{\bf times}}
\State Run $\mathcal{A}$ for~$2e^{\E\left[X\right] - \ln \E\left[X\right]}$ computational steps.
\EndFor
\For{$\lceil \ln\E[X]+2 \rceil$ \text{\bf times}}
\State Run $\mathcal{A}$ for~$2e^{\E\left[X\right] +2}$ computational steps.
\EndFor
\Until{$\mathcal{A}$ completed a run.}
\end{algorithmic}
\end{algorithm}

Due to Lemma~\ref{twothresholds}, each iteration of the outermost loop of Algorithm~\ref{alg2thresh} succeeds to fully execute~$\mathcal{A}$ with probability larger than~$1-e^{-1}$.
Thus, in expectation we run this loop for a constant number of iterations.
The first \textbf{for} loop takes~$O\left(\E\left[X\right] \cdot e^{\E\left[X\right] - \ln \E\left[X\right]}\right) = O\left(e^{\E\left[X\right]}\right)$ expected time, and the second takes~$O\left(e^{\E\left[X\right]} \ln \E\left[X\right]\right)$.
We conclude the following.
\begin{corollary}
Algorithm~\ref{alg2thresh} runs in expected time~$O\left(e^{\E\left[X\right]} \ln \E\left[X\right]\right)$. 
\end{corollary}

Intuitively, the \emph{proof} of Lemma~\ref{twothresholds} can be viewed as a reduction from the variable~$X$ to the variable~$Y|\left(Y>0\right)$, that has a much lower expectation: $\E\left[Y\;|\;Y>0\right]\leq \ln\E[X]+1$.
We can thus hope that iterating the proof for~$\ln^\star \E[X]$ times would result in reducing~$X$ to a variable with constant expectation.
A natural implementation of this idea would result in an algorithm that runs in expected time~$O\left(e^{\E\left[X\right]} \ln^\star \E\left[X\right]\right)$. 
We next formalize this intuition, and do so in a more careful manner to avoid the~$\ln^\star \E\left[X\right]$ factor. 

\begin{definition}
Let~$\lambda(x) := 3\ln(x)$ and note it is strictly increasing and shrinking for all~$x\geq 5$.
We define~$\lambda^\star(x)$, for~$x\geq 5$, to be the smallest~$k\in \mathbb{N}$ such that~$\lambda^{(k)}(x)\leq 5$.
\end{definition}

\begin{claim}\label{lambdabasic}
The following hold for all~$x\geq 5$:
\begin{enumerate}
    \item $\lambda^\star(x) = \Theta\left(\log^\star x\right)$.
    \item $\lambda^{\left(\lambda^\star\left(x\right)\right)}(x) > 4$.
    \item $\sum_{i=0}^{\lambda^\star(x)} \frac{1}{\lambda^{(i)}(x)} < 2$.
\end{enumerate}
\end{claim}
\begin{proof}
$(1)$ We have that~$\lambda^{(2)}(x) \leq \log x\leq \lambda(x)$ for all~$x\geq 410$.
In particular, $\log^\star x \leq \lambda^\star (x) \leq 2 \log^\star x + \lambda^\star (410)$.

$(2)$ $\lambda^{\left(\lambda^\star\left(x\right)-1\right)}(x) > 5$ and hence $\lambda^{\left(\lambda^\star\left(x\right)\right)}(x) > \lambda(5) > 4.82$.

$(3)$ For all~$x\geq 17$ it holds that~$\lambda(x)\leq \frac{x}{2}$.
Let~$k'$ be the smallest integer such that~$\lambda^{(k')}(x) < 17$.
We thus have $$\sum_{i=0}^{k'-1} \frac{1}{\lambda^{(i)}(x)} < \frac{1}{17} \sum_{i=0}^{\infty} 2^{-i} = \frac{2}{17}.$$
On the other hand, there are at most~$\lambda^\star(17)$ summands that are strictly larger than~$\frac{1}{17}$, thus by~$(2)$ we have
$$
\sum_{i=k'}^{\lambda^\star(x)} \frac{1}{\lambda^{(i)}(x)} < \frac{\lambda^\star(17)}{4} = \frac{5}{4}
.$$
\end{proof}


We are now ready to prove a generalized version of Lemma~\ref{twothresholds}, that is going to be the core of our final algorithm.

\begin{lemma}\label{corethresholds}
Let~$X$ be a non-negative distribution and~$E\geq \max\left(\E\left[X\right],\;5\right)$ be an upper bound on its expectation.
There either exists~$1\leq k \leq \lambda^\star(E)$ such that~$\Pr\left(X < E - \lambda^{(k)}\left(E\right)\right) \geq \left(\left(\lambda^{(k-1)}\left(E\right)+2\right)^2+1\right)^{-1}$, or it holds that $\Pr\left(X < E + 10\right) \geq \frac{1}{2}$.
\end{lemma}
\begin{proof}
We recursively denote by~$E_0:=E$ and by~$E_k:=\lambda^{(k)}(E) + \sum_{i=0}^{k-1}\frac{1}{E_i}$ for~$1\leq k \leq \lambda^\star (E)$. Note that~$E_k\geq \lambda^{(k)}(E)$ and hence also
$$
E_k = \lambda^{(k)}(E) + \sum_{i=0}^{k-1}\frac{1}{E_i}
\leq \lambda^{(k)}(E) + \sum_{i=0}^{k-1}\frac{1}{\lambda^{(i)}(E)}
< \lambda^{(k)}(E) +2,
$$
where the last inequality follows from Claim~\ref{lambdabasic}.
In particular, $\lambda^{(k)}(E) \leq E_k < \lambda^{(k)}(E) +2$.

Assume that~$\Pr\left(X < E - \lambda^{(k)}\left(E\right)\right) < \left(\left(\lambda^{(k-1)}\left(E\right)+2\right)^2+1\right)^{-1} < \frac{1}{\left(E_{k-1}\right)^2+1}$ for every~$1\leq k \leq \lambda^\star\left(E\right)$.

Denote by~$Y_k := X-\left(E-\lambda^{(k)}\left(E\right)\right)$ for~$k\geq 0$.
We prove by induction on~$k$ that~$\E\left[Y_k \; | \; Y_k \geq 0\right]\leq E_k$.
For~$k=0$ the claim is straightforward as~$Y_0 = X$ and~$E_0=E$.
For the inductive step, we assume the hypothesis holds for~$k-1$ and show it holds for~$k$.
We note that~$Y_{k-1}>Y_k$  and hence if~$Y_k\geq 0$ then~$Y_{k-1}\geq 0$ as well. Hence,
\begin{align*}
\E[Y_{k-1}\;|\;Y_{k-1}\geq 0] &\geq \Pr\left(Y_{k} \geq 0\;|\;Y_{k-1}\geq 0\right) \E\left[Y_{k-1} \; | \; Y_k \geq 0\right]\\
&\geq
\Pr\left(Y_{k} \geq 0\right) \E\left[Y_{k-1} \; | \; Y_k \geq 0\right].
\end{align*}
Thus, by the induction hypothesis we have
\begin{align*}
\E\left[Y_{k-1} \; | \; Y_k \geq 0\right] &\leq \frac{\E[Y_{k-1}\;|\;Y_{k-1}\geq 0]}{\Pr\left(Y_{k} \geq 0\right)}\\
&\leq 
\frac{E_{k-1}}{1 - \frac{1}{\left(E_{k-1}\right)^2+1}}
\\&=
E_{k-1} + \frac{1}{E_{k-1}}
.
\end{align*}
Therefore, 
\begin{align*}
    \E[Y_{k}\;|\;Y_{k}\geq 0] &= \E[Y_{k-1}\;|\;Y_{k}\geq 0] + \lambda^{(k)}\left(E\right) - \lambda^{(k-1)}\left(E\right) \\&\leq E_{k-1} + \frac{1}{E_{k-1}} + \lambda^{(k)}\left(E\right) - \lambda^{(k-1)}\left(E\right)\\&=E_k.
\end{align*}

In particular, we have that~$\E\left[Y_{\lambda^\star\left(E\right)} \; | \; Y_{\lambda^\star\left(E\right)} \geq 0\right]\leq E_{\lambda^\star\left(E\right)} < \lambda^{\left(\lambda^\star\left(E\right)\right)}(E) + 2 \leq 7$.
Therefore,
\begin{align*}
\Pr\left(X \geq E + 10\right) &\leq \Pr\left(X \geq E + 10\; \bigg\rvert \; X \geq E-\lambda^{\left(\lambda^\star\left(E\right)\right)}(E)\right) \\
&= \Pr\left(Y_{\lambda^\star\left(E\right)} \geq \lambda^{\left(\lambda^\star\left(E\right)\right)}(E) + 10\; \bigg\rvert \; Y_{\lambda^\star\left(E\right)} \geq 0\right) \\
&\leq \Pr\left(Y_{\lambda^\star\left(E\right)} \geq 14\; \bigg\rvert \; Y_{\lambda^\star\left(E\right)} \geq 0\right) < \frac{7}{14} = \frac{1}{2}.
\end{align*}
\end{proof}

Consider the following algorithm.
\begin{algorithm}[H]
\caption{}\label{algspecificE}
\begin{algorithmic}[1]
\Require{$\mathcal{A}$, $E$.}
\Repeat
\For{$k=1$ to $\lambda^\star(E)$}
\For{$2\lceil \left(\lambda^{(k-1)}\left(E\right)+2\right)^2+1 \rceil$ \text{\bf times}}
\State Run $\mathcal{A}$ for~$2e^{E - \lambda^{(k)}\left(E\right)}$ computational steps.
\EndFor
\EndFor
\For{$2$ \text{\bf times}}
\State Run $\mathcal{A}$ for~$2e^{E + 10}$ computational steps.
\EndFor
\Until{$\mathcal{A}$ completed a run.}
\end{algorithmic}
\end{algorithm}

\begin{corollary}[of Lemma~\ref{corethresholds}]\label{coresuccessprob}
Each \textbf{repeat} loop of Algorithm~\ref{algspecificE} fully executes~$\mathcal{A}$ with probability at least~$\frac{3}{4}$.
\end{corollary}

\begin{lemma}\label{iterationruntime}
Let~$E\geq \max\left(\E[X],\; 5\right)$, Algorithm~\ref{algspecificE} runs in~$O\left(e^E\right)$ expected time.
\end{lemma}
\begin{proof}
By Corollary~\ref{coresuccessprob} we enter the \textbf{repeat} loop a constant number of times in expectation. We thus analyze the computational cost of a single \textbf{repeat} loop.
The evaluations in Lines~$5-6$ take~$O\left(e^E\right)$ time. 
The evaluations in Lines~$2-4$ take
$$
\sum_{k=1}^{\lambda^\star (E)} 2\lceil \left(\lambda^{(k-1)}\left(E\right)+2\right)^2+1 \rceil \cdot 2e^{E - \lambda^{(k)}\left(E\right)}
=
O\left(e^E \cdot \sum_{k=1}^{\lambda^\star (E)} \left(\lambda^{(k-1)}\left(E\right)\right)^2 e^{- \lambda^{(k)}\left(E\right)}\right)
$$
time.
By the definition of~$\lambda(x)$, we have~$e^{-\lambda^{(k)}(x)}=e^{-3\ln\left(\lambda^{(k-1)}\left(x\right)\right)}=\left(\lambda^{(k-1)}\left(x\right)\right)^{-3}$.
In particular,
$$
\sum_{k=1}^{\lambda^\star (E)} \left(\lambda^{(k-1)}\left(E\right)\right)^2 e^{- \lambda^{(k)}\left(E\right)}
=
\sum_{k=1}^{\lambda^\star (E)} \left(\lambda^{(k-1)}\left(E\right)\right)^{-1}
< 2
,$$
where the last inequality follows from Claim~\ref{lambdabasic}.
\end{proof}

Finally, we also get rid of the necessity to provide the algorithm with~$E$ or~$\E[X]$.

\begin{algorithm}[H]
\caption{}\label{algfinal}
\begin{algorithmic}[1]
\Require{$\mathcal{A}$.}
\For{$E=5$ to $\infty$}
\For{$k=1$ to $\lambda^\star(E)$}
\For{$2\lceil \left(\lambda^{(k-1)}\left(E\right)+2\right)^2+1 \rceil$ \text{\bf times}}
\State Run $\mathcal{A}$ for~$2e^{E - \lambda^{(k)}\left(E\right)}$ computational steps.
\EndFor
\EndFor
\For{$2$ \text{\bf times}}
\State Run $\mathcal{A}$ for~$2e^{E + 10}$ computational steps.
\EndFor
\State\Return{\textbf{if} $\mathcal{A}$ completed a run.}
\EndFor
\end{algorithmic}
\end{algorithm}

\begin{theorem}
Algorithm~\ref{algfinal} runs in expected time~$O\left(e^{\E[X]}\right)$.
\end{theorem}
\begin{proof}
By Lemma~\ref{iterationruntime} the iteration of the outermost \textbf{for} loop corresponding to~$E$ takes at most~$C\cdot e^E$ time, for some global constant~$C$.
All iterations in which~$E<\E[X]$ thus take~$O\left(e^{\E[X]}\right)$ time.
By Corollary~\ref{coresuccessprob}, each subsequent iteration succeeds with probability at least~$\frac{3}{4}$.
Thus the expected running time is bounded by
$$
Ce^{\E[X]} \cdot \sum_{t=0}^{\infty} e^t \left(\frac{1}{4}\right)^{t}
= O\left(e^{\E[X]}\right)
.$$
\end{proof}

\section{Conclusions and Open Problems}\label{sec:conclude}
We showed that a Las-Vegas algorithm with expected running~$e^X$ conditioned on the value of some random variable~$X$, can always be converted to a Las-Vegas algorithm with expected running time~$O\left(e^{\E[X]}\right)$.
In particular, an algorithm whose running time is a random variable~$T$ can be converted to one with expected running time~$O\left(e^{\E[\ln T]}\right)$, which is never worse than~$O(\E[T])$.

We demonstrated a use of this theorem to simplify a proof in the regime of exponential time algorithms~\cite{zamir2021faster}.
It is interesting to try applying it to other exponential and non-exponential time algorithms and see if it can simplify or even improve the analysis.

\subsection{Considering the variance}
In terms of~$\E[X]$ only, we can not get any better than~$O\left(e^{\E[X]}\right)$ as the distribution of~$X$ might be constant.
In that case though, the variance of~$X$ is zero.
Can we get a better bound just by assuming that the variance of~$X$ is large?
Unfortunately, with the standard definition of variance this is not the case.
For any choice of~$E$ and~$V\geq 2E^2 e^{-E}$ consider the following distribution:
\[
\Pr\left(X=k\right) :=
\begin{cases} 
      e^{-E} & k=0 \\
      1-\frac{Ve^{-E}}{V-E^2 e^{-E}} & k=E \\
      \frac{\left(Ee^{-E}\right)^2}{V-E^2 e^{-E}} & k=\frac{V}{Ee^{-E}}
   \end{cases}.
\]
Its expectation is~$E$, its variance is~$V$, which can be arbitrarily large, and nevertheless~$\Pr(X<E)=e^{-E}$ so no strategy can beat~$O\left(e^E\right)$.

On the other hand, the wishful thinking above is true with some other notions of variance.
For example, if we consider mean deviation instead of standard deviation (i.e., $\E\left[\;|X-\E\left[X\right]|\;\right]$), then \emph{it is} true that if the variance is large then we can get a better running time.
It is intriguing to find useful notion of variance for which such a statement is true, with the goal of improving the running time of algorithms by analyzing the variance of~$X$.

In particular, consider the PPSZ algorithm for solving~$k$-SAT~\cite{paturi2005improved}~\cite{hertli20143}.
The algorithm uses randomization in two ways: first, a random permutation of the variables in the input formulas is drawn; Then, the chosen permutation determines the number of variables we need to \emph{guess} the value of.
In a recent improvement of the PPSZ analysis, Scheder~\cite{scheder2022ppsz} showed that in some large subset of permutations the number of guessed variables is smaller than what we expect when taking a uniformly random permutation.
In particular, this implies that there is some non-negligible variance in the original algorithm's running time.
Can we get better SAT algorithms by analyzing this variance?

\subsection*{Acknowledgements}
The author would like to thank Avi Wigderson for pointing out important references.

\bibliographystyle{alpha}
\bibliography{main.bib}

\newcommand{\etalchar}[1]{$^{#1}$}
\begin{thebibliography}{AGM{\etalchar{+}}96}

\bibitem[AB09]{arora2009computational}
Sanjeev Arora and Boaz Barak.
\newblock {\em Computational complexity: a modern approach}.
\newblock Cambridge University Press, 2009.

\bibitem[AGM{\etalchar{+}}96]{alt1996method}
Helmut Alt, Leonidas Guibas, Kurt Mehlhorn, Richard Karp, and Avi Wigderson.
\newblock A method for obtaining randomized algorithms with small tail
  probabilities.
\newblock {\em Algorithmica}, 16(4):543--547, 1996.

\bibitem[Her14]{hertli20143}
Timon Hertli.
\newblock 3-sat faster and simpler---unique-sat bounds for ppsz hold in
  general.
\newblock {\em SIAM Journal on Computing}, 43(2):718--729, 2014.

\bibitem[HKZZ19]{hansen2019faster}
Thomas~Dueholm Hansen, Haim Kaplan, Or~Zamir, and Uri Zwick.
\newblock Faster k-sat algorithms using biased-ppsz.
\newblock In {\em Proceedings of the 51st Annual ACM SIGACT Symposium on Theory
  of Computing}, pages 578--589, 2019.

\bibitem[LSZ93]{luby1993optimal}
Michael Luby, Alistair Sinclair, and David Zuckerman.
\newblock Optimal speedup of las vegas algorithms.
\newblock {\em Information Processing Letters}, 47(4):173--180, 1993.

\bibitem[PPSZ05]{paturi2005improved}
Ramamohan Paturi, Pavel Pudl{\'a}k, Michael~E Saks, and Francis Zane.
\newblock An improved exponential-time algorithm for k-sat.
\newblock {\em Journal of the ACM (JACM)}, 52(3):337--364, 2005.

\bibitem[Sch22]{scheder2022ppsz}
Dominik Scheder.
\newblock Ppsz is better than you think.
\newblock In {\em 2021 IEEE 62nd Annual Symposium on Foundations of Computer
  Science (FOCS)}, pages 205--216. IEEE, 2022.

\bibitem[Zam22]{zamir2021faster}
Or~Zamir.
\newblock Faster algorithm for unique $(k, 2) $-csp.
\newblock {\em ESA}, 2022.

\end{thebibliography}

\end{document}